\documentclass[10pt,twocolumn,twoside]{IEEEtran}
\usepackage{cite}
\usepackage{amsmath,amssymb,amsfonts}
\usepackage{textcomp}
\usepackage{verbatim}
\usepackage{titletoc}
\usepackage{caption}
\usepackage{subfigure}
\usepackage[ruled,vlined]{algorithm2e}
\usepackage{enumerate}
\usepackage[hidelinks]{hyperref} 
\usepackage{wrapfig}
\usepackage{graphicx}
\usepackage{subfigure}
\usepackage{bm} 
\usepackage{amsthm}
\usepackage{array}
\usepackage{xcolor}
\usepackage{longtable}
\usepackage{multirow}
\usepackage{multicol}
\usepackage{fancyhdr}
\usepackage{setspace}
\usepackage{booktabs}
\usepackage{xcolor}
\usepackage{algorithmicx,algpseudocode}
\usepackage{mathrsfs}
\usepackage[numbers,sort&compress]{natbib}
\usepackage{pifont}

\allowdisplaybreaks[4]
\theoremstyle{definition}
\newtheorem{thm}{Theorem}[]
\newtheorem{lemma}[thm]{Lemma}

\newtheorem{remark}[thm]{Remark}

\newtheorem{definition}[thm]{Definition}
\newtheorem{prop}[thm]{Proposition}

\newcolumntype{C}[1]{>{\centering\let\newline\\\arraybackslash\hspace{0pt}}m{#1}}
\def\BibTeX{{\rm B\kern-.05em{\sc i\kern-.025em b}\kern-.08em
   T\kern-.1667em\lower.7ex\hbox{E}\kern-.125emX}}
\newcommand{\cmark}{\ding{51}}%
\newcommand{\xmark}{\ding{55}}%
    
\makeatletter
\newcommand*\bigcdot{\mathpalette\bigcdot@{.5}}
\newcommand*\bigcdot@[2]{\mathbin{\vcenter{\hbox{\scalebox{#2}{$\m@th#1\bullet$}}}}}
\makeatother

\begin{document}
\title{Control Data Scheduling over Shared Communication Channels: A {Sparse} and Collision-Free Mechanism}
\author{Yuxing Zhong, Zhaohua Yang, Fuhai Nan, Daniel E.~Quevedo, and Ling Shi
\thanks{Yuxing Zhong was with the Department of Electronic and Computer Engineering, Hong Kong University of Science and Technology when this work was conducted, and is now with the School of Electrical and Computer Engineering, University of Sydney, Australia (email: \href{yuxing.zhong@connect.ust.hk}{yuxing.zhong@connect.ust.hk}).}
\thanks{Fuhai Nan, Zhaohua Yang and Ling Shi are with the Department of Electronic and Computer Engineering, Hong Kong University of Science and Technology, Hong Kong (e-mail: \href{zyangcr@connect.ust.hk}{zyangcr@connect.ust.hk}; \href{fnan@connect.ust.hk}{fnan@connect.ust.hk}; \href{eesling@ust.hk}{eesling@ust.hk}). Ling Shi is also with the Department of Chemical and Biological Engineering, Hong Kong University of Science and Technology.}
\thanks{Daniel E.~Quevedo is with the School of Electrical and Computer Engineering, University of Sydney, Australia (email:\href{daniel.quevedo@sydney.edu.au}{daniel.quevedo@sydney.edu.au}).}}
\maketitle

\begin{abstract}
In systems where controllers operate remotely and communicate with actuators over shared communication channels, it is crucial to efficiently schedule the control data transmission to reduce bandwidth usage and actuator effort. In this article, we investigate a novel scheduling mechanism that jointly coordinates control actions across different controllers and different time steps. We propose an algorithm based on the alternating direction method of multipliers (ADMM) to solve the resulting optimization problem. While ADMM is often treated as a black-box solver, the proposed algorithm offers a clear physical interpretation, ensures convergence to a stationary point, and is computationally efficient. Simulation results validate our theoretical results and demonstrate the effectiveness of our proposed algorithm.
\end{abstract}

\section{Introduction}
Recent advancements in networked control systems (NCSs) have facilitated wireless sensing, estimation, and control~\cite{hespanha2007survey,walsh2001scheduling}. For example, controllers can now be deployed remotely, with control data transmitted to actuators via wireless communications. While this architecture reduces system wiring, thus enhancing system flexibility and lowering operational costs, it also introduces significant challenges. The limitation of communication bandwidth, coupled with the need to minimize actuator effort, necessitates the careful scheduling of control data transmission to optimize control performance within these resource constraints.

\subsubsection{Scheduling over Time}
Scheduling-over-time methods have been proposed to minimize actuator effort during the operational period. By designing sparse control data, actuators can remain inactive for the majority of the time. This is particularly important to systems aimed at conserving energy while reducing noise and vibration~\cite{chan2007state}. Building on this idea, Shi et al.~\cite{shi2012finite} and Nishida and Okano~\cite{nishida2024sparsity} developed scheduling strategies to minimize the quadratic cost while 
limiting actuator activation times. {Siami and Jadbabaie~\cite{siami2020separation} later proposed scheduling algorithms based on Hankel-based performance metrics.} Furthermore, Kishida~\cite{kishida2018hands} and {Nagahara et al.~\cite{nagahara2015maximum}} introduced a novel maximum hands-off design, which minimizes the number of non-zero control signals while satisfying performance or terminal-state requirements. However, these approaches are restricted to systems with a single controller-actuator pair. As modern systems increasingly incorporate multiple controllers and actuators, their applicability to more complex configurations becomes significantly constrained.

\subsubsection{Scheduling over Controllers}
In systems with multiple controllers and actuators communicating over shared networks, {simultaneous transmissions can result in packet collisions, causing all transmissions to fail}. Therefore, selecting a subset of controllers for transmission is essential to avoid communication collisions~\cite{zhong2024sparse}. While extensive research on control data scheduling has primarily focused on analyzing the controllability or stability of the system under specific scheduling strategies~\cite{pasqualetti2014controllability,summers2014optimal,siami2020deterministic,ballotta2024pointwise,joseph2020controllability,pasand2017structural}, the effectiveness of these approaches remains unclear. Jiao et al.~\cite{jiao2022actuator} addressed this issue by minimizing the quadratic cost of the system. However, their method restricts the system to a single active actuator at each time instant.

\begin{table*}[!htbp]
	\small
	\centering
	\caption{Overview of existing literature, where `S' and `M' denote `single' and `multiple', respectively.}
	
	\label{tb:summary}
	\begin{tabular}{ll|cc|cc|l}
		\hline
		\multicolumn{2}{l|}{}                                                                                                                                              & \multicolumn{2}{c|}{System Model} &                          &                                           &                                                            \\\cline{3-4}
		\multicolumn{2}{l|}{\multirow{-2}{*}{}}                                                                                                                            & Plant      & Actuator  & \multirow{-2}{*}{\begin{tabular}[c]{@{}c@{}}Vector\\  Form\end{tabular}} & \multirow{-2}{*}{\begin{tabular}[c]{@{}c@{}}Active\\  Actuators\end{tabular}} & \multirow{-2}{*}{Objective}                                \\ \hline
		\multicolumn{1}{l|}{}                                   & Shi et   al.~\cite{shi2012finite}                          & S     & S               & \cmark                      & NA                                        & Minimize the  quadratic cost                                  \\
		\multicolumn{1}{l|}{}                                   & Nishida and   Okano~\cite{nishida2024sparsity}             & S     & S               & \cmark                      & NA                                        & Minimize the quadratic cost                                    \\
		\multicolumn{1}{l|}{}                                   & Siami and   Jadbabaie~\cite{siami2020separation}                                   & S     & S               & \cmark                      & NA                                        & {Optimize Hankel-based   performance} \\ 
		\multicolumn{1}{l|}{}                                   & Kishida~\cite{kishida2018hands}                            & S     & S               & \cmark                      & NA                                        & Minimize the quadratic cost                                     \\
		\multicolumn{1}{l|}{\multirow{-5}{*}{\begin{tabular}[c]{@{}l@{}}Over\\ Time\end{tabular}}}                                   & Nagahara et   al.~\cite{nagahara2015maximum}               & S     & M             & \cmark                      & NA                                        & Steer the state to a target   value                        \\\hline
		\multicolumn{1}{l|}{}                                   & Pasqualetti et   al.~\cite{pasqualetti2014controllability} & S     & M             & \cmark                      & M                                       & {Steer the state to a target   value} \\
		\multicolumn{1}{l|}{}                                   &Summers and   Lygeros~\cite{summers2014optimal}            & S     & M             & \cmark                      & M                                       & Optimize   controllability  metrics                        \\
		\multicolumn{1}{l|}{}                                   & Siami et   al.~\cite{siami2020deterministic}                                       & S     & M             & \xmark                       & M                                       & Optimize   controllability metrics                         \\
		\multicolumn{1}{l|}{}                                   & Ballotta et   al.~\cite{ballotta2024pointwise}                                     & S     & M             & \xmark                       & M                                       & Optimize   controllability metrics                         \\
		\multicolumn{1}{l|}{}                                   & Joseph and   Murthy~\cite{joseph2020controllability}                               & S     & M             & \xmark                       & M                                       & Analyze   controllability conditions                       \\
		\multicolumn{1}{l|}{}       & Pasand and   Montazeri~\cite{pasand2017structural}                                 & S     & M             & \xmark                       & M                                       & Analyze   controllability conditions                       \\
		\multicolumn{1}{l|}{\multirow{-7}{*}{\begin{tabular}[c]{@{}l@{}}Over\\ Controllers\end{tabular}}}        & Jiao et   al.~\cite{jiao2022actuator}                      & S     & M             & \cmark                      & S                                        & Minimize   quadratic cost                                   \\ \hline
		\multicolumn{1}{l|}{}                                   & Ikeda and   Kashima~\cite{ikeda2022sparse}                                         & S     & M             & \xmark                       & M                                       & Optimize   controllability metrics                         \\
		\multicolumn{1}{l|}{}                                   & Olshevsky~\cite{olshevsky2020relaxation}                                           & S     & M             & \xmark                       & M                                       & Analyze linear programming (LP) relaxation                     \\
		\multicolumn{1}{l|}{}                                   & Xu et   al.~\cite{xu2018optimal}                                                   & M   & M             & \xmark                       & M                                       & Minimize the  quadratic cost                                   \\
		\multicolumn{1}{l|}{\multirow{-4}{*}{\begin{tabular}[c]{@{}l@{}}Over\\ Both Axes\end{tabular}}}             & {\bf Ours}                                                                                                     & {\bf M}   & {\bf M}             & \cmark                      & {\bf M}                                       & {\bf Minimize the quadratic cost}                                   \\ \hline
	\end{tabular}
\end{table*}

\begin{figure}[!tbp]
	\centering
	\includegraphics[width=\linewidth]{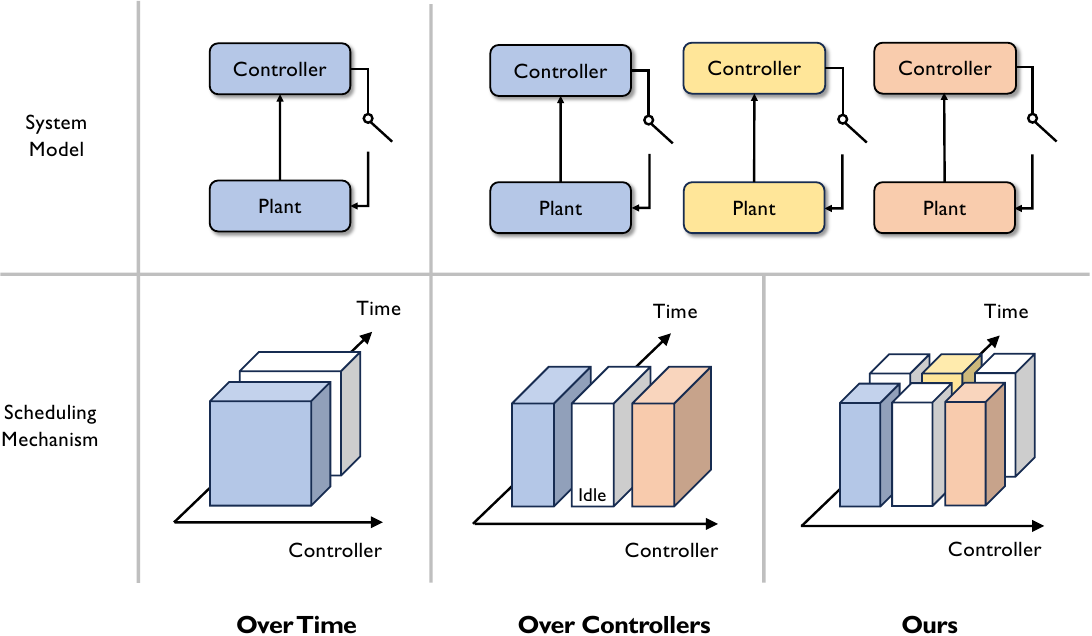}
	\caption{Overview of control data scheduling algorithms.}
	\label{fig:literature}
\end{figure}
Fig.~\ref{fig:literature} and {Table~\ref{tb:summary}} summarize the above existing works. While these works are valuable, they only focus on designing scheduling algorithms on one side, i.e., either over time or over controllers. {However, in many modern applications~\cite{ma2022smart,mozaffari2019tutorial}, controllers are deployed remotely with substantial computational resources, whereas actuators are battery-powered and operate under energy constraints. In such settings, it is desirable to reduce control effort while maintaining collision avoidance, even at the expense of increased controller-side computation. Achieving this objective requires joint scheduling over both axes. This problem is more challenging due to the stronger spatiotemporal coupling of control signals. Although several studies have considered joint scheduling~\cite{ikeda2022sparse,olshevsky2020relaxation,xu2018optimal}, their results are limited to scalar systems. In contrast, we investigate joint control data scheduling for general vector systems.} Our main contributions are:
\begin{enumerate}
\item We formulate a novel joint scheduling problem over both controllers and time. Unlike existing works~\cite{ikeda2022sparse,olshevsky2020relaxation,xu2018optimal}, our framework is applicable to general vector systems.
\item We develop an ADMM-based algorithm for the resulting problem. {Although ADMM is widely used in optimization, establishing its convergence or solution quality for general non-convex problems remains difficult~\cite{magnusson2015convergence}. Unlike most existing works that use ADMM solely as a computational tool}, our algorithm admits a clear physical interpretation ({Fig.~\ref{fig:admm}}) and guaranteed convergence ({Theorem~\ref{thm:converge1}} and {Theorem~\ref{thm:converge2}}) to a stationary point ({Theorem~\ref{thm:stationary}}).
\end{enumerate}

\emph{Notations:} The notations $\mathbb{R}$ and $\mathbb{R}^n$ represent the sets of real numbers and $n$-dimensional vectors, respectively. For a matrix $\bm A$, $\bm A'$ denotes its transpose, $\|\bm A\|_F$ is its Frobenius norm, and $\bm A_{i,\bigcdot}$ and $\bm A_{\bigcdot,j}$ correspond to its $i$-th row and $j$-th column vector, respectively. The notation $\bm A\succeq0$ means that $\bm A$ is positive semidefinite. The operators $\otimes$, $\odot$, and $\oslash$ refer to the Kronecker product, Hadamard product, and Hadamard division, respectively. The notations $\|\cdot\|_0$ and $\|\cdot\|_2$ denote the $\ell_0$ and $\ell_2$ norms, while $|\cdot|$ is the absolute value. The expression $\bm I$ is the identity matrix. Bold values such as $\bm 0,\bm 1,\dots$ represent matrices or vectors of the corresponding values with compatible dimensions. The function $\rm{vec}(\cdot)$ returns the vectorized form of a matrix, and $\rm{sign}(x)$ is the sign function, i.e., $\rm{sign}(x)=-1$ if $x<0$; $\rm{sign}(x)=1$ if $x>0$; and $\rm{sign}(x)=0$ if $x=0$.

\section{Problem Formulation}\label{chap:formulation}
\subsection{System Model}
Consider $N$ discrete linear time-invariant systems over a finite time horizon $T$ (Fig.~\ref{fig:system}). The dynamics of $i$-th system are given by
\begin{equation}\label{eq:dynamics}
{\bm x}_{k+1,i} = {\bm A}_i{\bm x}_{k,i} + \zeta_{k,i}{\bm B}_i{\bm u}_{k,i},\quad i\in\mathcal{N}, \quad k\in\mathcal{T},
\end{equation}
where $\mathcal{N}=\{1,\dots,N\}$, $\mathcal{T}=\{0,\dots,T-1\}$, ${\bm x}_{k,i}\in\mathbb{R}^{n_i}$ is the system state, ${\bm u}_{k,i}\in\mathbb{R}^{m_i}$ is the control input, $\zeta_{k,i}\in\{0,1\}$ is the binary decision variable indicating whether ${\bm u}_{k,i}$ is transmitted to the plant, i.e., $\zeta_{k,i}=1$ if ${\bm u}_{k,i}$ is transmitted to the plant and $\zeta_{k,i}=0$ otherwise. {We assume that transmissions are instantaneous and error-free. Therefore, when $\zeta_{k,i}=1$, the plant receives the control signal ${\bm u}_{k,i}$ perfectly without delay.} 
{
\begin{remark}{\rm
	This assumption is commonly adopted in the existing literature~\cite{shi2012finite,kishida2018hands,jiao2022actuator,nishida2024sparsity} and allows us to focus on the joint scheduling design considered in this paper. Incorporating delays or packet losses would introduce randomness into the system, thereby transforming the deterministic optimization problem in~\eqref{eq:pb} into a stochastic one. In such settings, analytical~\cite{zhong2024sparse} or sampling-based~\cite{mesbah2016stochastic} approaches can be employed. We leave these extensions as our future work. 	
}\end{remark}}

For $i\in\mathcal{N}$ and $k\in\mathcal{T}$, we seek control signals ${\bm u}_{k,i}$ and scheduling strategies $\zeta_{k,i}$ that minimize the finite horizon quadratic cost $J=\sum_{i=1}^N J_i$, where
\begin{equation*}
\begin{aligned}
J_i=\sum_{k=0}^{T-1}[{\bm x}_{k,i}'{\bm Q}_i{\bm x}_{k,i}+\zeta_{k,i}{\bm u}_{k,i}'{\bm R}_i{\bm u}_{k,i}]+{\bm x}'_{T,i}{\bm Q}_i{\bm x}_{T,i},
\end{aligned}
\end{equation*}
with ${\bm Q}_i\succeq0$ and ${\bm R}_i\succ0$.
\begin{figure}[!htbp]
	\centering
	\includegraphics[width=\linewidth]{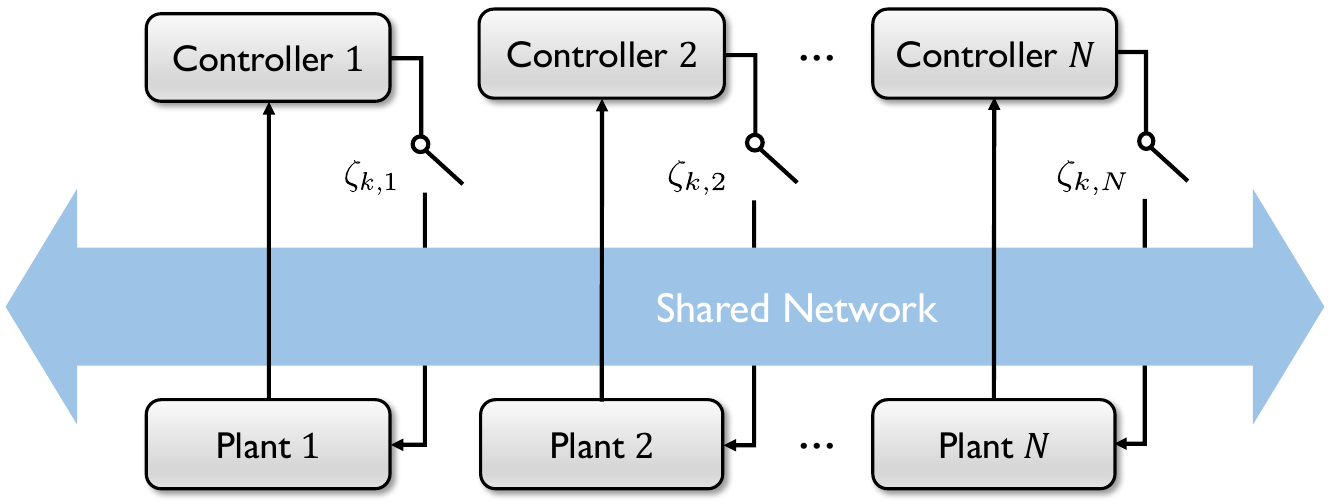}
	\caption{System model.}
	\label{fig:system}
\end{figure}

As mentioned, we have requirements on $\zeta_{k,i}$ both over controllers and over time. Specifically:
\begin{enumerate}
\item {\bf Over Controllers}: To avoid communication collisions, only a subset of controllers is allowed to transmit simultaneously, i.e., $\sum_{i=1}^N\zeta_{k,i}\leq z$ for each time $k$, where $z< N$ is a positive integer.
\item {\bf Over Time}: {For each system $i$, we enforce sparsity on $\bm{\zeta}_i\triangleq[\zeta_{0,i},\dots,\zeta_{T-1,i}]$ to reduce the actuator effort over the operational horizon.}
\end{enumerate}

Our problem of interest can then be formulated as the following optimization problem:
\begin{equation}\label{eq:pb}
\begin{aligned}
\min_{\{\zeta_{k,i}\}, \{{\bm u}_{k,i}\}} 	&\quad \sum_{i=1}^N\left\{J_i+\alpha_i {\|\bm{\zeta}_i\|_0} \right\}\\
{\rm s.t.}					&\quad \sum_{i=1}^N\zeta_{k,i}\leq z,\quad k\in\mathcal{T},\\
						&\quad \zeta_{k,i}\in\{0,1\}\text{ and \eqref{eq:dynamics}},\quad k\in\mathcal{T},\quad i\in\mathcal{N},
\end{aligned}
\end{equation}
where $\alpha_i\geq0$ is the penalty weight representing the level of emphasis on the sparsity of $\bm{\zeta}_i$.

{
\begin{remark}
	Building upon the standard linear quadratic regulator (LQR) framework~\cite{bertsekas2012dynamic}, we do not impose additional constraints on $\bm{x}_{k,i}$ and $\bm{u}_{k,i}$. Nevertheless, it is important to note that the proposed ADMM framework can be readily extended to accommodate convex constraints in linear or quadratic form. For instance, for the commonly adopted box constraints, both the closed-form updates and the convergence guarantees remain valid.
\end{remark}}
\subsection{Reformulation}\label{sec:reformulation}
Solving~\eqref{eq:pb} directly is challenging due to its nature as an {NP-hard} mixed-integer program (MIP)~\cite{karp2010reducibility}. Before we develop the optimization algorithm in the following section, let us first reformulate~\eqref{eq:pb} into an equivalent form {with respect to scheduling strategies.}

Denote $\tilde{\bm u}_{k,i}\triangleq\zeta_{k,i}{\bm u}_{k,i}$ and define the following:
\begin{align*}
\bar{\bm x}_i&\triangleq\begin{bmatrix}{\bm x}_{0,i}\\ \vdots\\ {\bm x}_{T,i}\end{bmatrix},\quad \bar{\bm u}_i\triangleq\begin{bmatrix}\tilde{\bm u}_{0,i}\\ \vdots\\ \tilde{\bm u}_{T-1,i}\end{bmatrix},
\begin{array}{ll}
\bar{\bm Q}_i\triangleq {\bm I}_{T+1}\otimes {\bm Q}_i, \vspace{0.1cm}\\ \vspace{0.1cm}
\bar{\bm R}_i\triangleq {\bm I}_T\otimes {\bm R}_i,
\end{array}
\\
\bar{\bm A}_i&\triangleq
\begin{bmatrix}
{\bm I}_{n_i}\\
{\bm A}_i\\
\vdots\\
{\bm A}_i^{T}
\end{bmatrix},\quad
\bar{\bm B}_i\triangleq
\begin{bmatrix}
\bm{0}		&\bm{0}&\cdots&\bm{0}\\
{\bm B_i}	&\bm{0}&\cdots	&\bm{0}\\
\bm{A}_i\bm{B}_i	&\bm{B}_i&~	&\bm{0}\\
\vdots		&\vdots&\ddots	&~\\
\bm{A}_i^{T-1}\bm{B}_i	&\bm{A}_i^{T-2}\bm{B}_i&\cdots	&{\bm B}_i\\
\end{bmatrix}.
\end{align*}
Then, for each system $i\in\mathcal{N}$, ~\eqref{eq:dynamics} and $J_i$ can be rewritten as
\begin{align*}
\bar{\bm x}_i	&=\bar{\bm A}_i{\bm x}_{0,i}+\bar{\bm B}_i\bar{\bm u}_i,\\
 J_i		&=\bar{\bm x}_i'\bar{\bm Q}_i\bar{\bm x}_i + \bar{\bm u}_i' \bar{\bm R}_i\bar{\bm u}_i=\bar{\bm u}_i'{\bm P}_i\bar{\bm u}_i + {\bm q}_i'\bar{\bm u}_i+{r}_i,
\end{align*}
where ${\bm P}_i=\bar{\bm B}_i'\bar{\bm Q}_i\bar{\bm B}_i+\bar{\bm R}_i$, ${\bm q}_i= 2\bar{\bm B}_i'\bar{\bm Q}_i\bar{\bm A}_i    {\bm x}_{0,i}$ and $r_i =\bm{x}_{0,i}'\bar{\bm A}_i'\bar{\bm Q}\bar{\bm A}_i{\bm x}_{0,i}$. Therefore,~\eqref{eq:pb} can be transformed into
\begin{equation}\label{eq:u}
\begin{aligned}
\min_{\{\tilde{\bm u}_{k,i}\}}	 		&\quad \sum_{i=1}^N\left[\left(\bar{\bm u}_i'{\bm P}_i\bar{\bm u}_i + {\bm q}_i'\bar{\bm u}_i\right)+\alpha_i\sum_{k=0}^{T-1}\mathbb{I}\left(\|\tilde{\bm u}_{k,i}\|_0\neq0\right)\right]\\
{\rm s.t.}					&\quad\sum_{i}^N\mathbb{I}\left(\|\tilde{\bm u}_{k,i}\|_0\neq0\right)\leq z,\quad k\in\mathcal{T},					
\end{aligned}
\end{equation}
where $\mathbb{I}(\cdot)$ denotes the indicator function, which returns $1$ when the condition is satisfied and $0$ otherwise.
{
\begin{remark}
	Problem~\eqref{eq:pb} and~\eqref{eq:u} are equivalent with respect to objective value and the resulting scheduling decisions. Specifically, when $\zeta_{k,i}=0$, the value of $\bm{u}_{k,i}$ does not have any impact on system dynamics. Therefore, any feasible solution to~\eqref{eq:pb} can be transformed into an equivalent feasible solution with the same objective value by setting $\bm{u}_{k,i}=0$ whenever $\zeta_{k,i}=0$.
\end{remark}}

\section{Methodology}\label{sec:alg}
Without loss of generality, we consider the case $m_i = 1$ in this section to simplify the presentation. The proposed algorithm can be readily generalized to arbitrary $m_i$ via straightforward matrix manipulations. 	Note that in the simulation part (Section~\ref{chap:sim}), we consider general systems with $m_i$ not necessarily equal to $1$.

Define ${\bm U}\triangleq[\bar{\bm u}_1,\dots,\bar{\bm u}_N]\in\mathbb{R}^{T\times N}$, ${\bm q}\triangleq[{\bm q}_1',\dots,{\bm q}_N']'\in\mathbb{R}^{TN}$ and ${\bm P}\triangleq{\rm blkdiag}({\bm P}_1,\dots,{\bm P}_N)\in\mathbb{R}^{TN\times TN}$. Using these definitions, we can rewrite~\eqref{eq:u} {equivalently} as: 
\begin{equation}\label{eq:U}
\begin{aligned}
\min_{\bm U}	 		&\quad {\rm vec}(\bm U)'{\bm P}{\rm vec}(\bm U) + {\bm q}'{\rm vec}(\bm U)+\sum_{i=1}^N \alpha_i\|{\bm U}_{\bigcdot,i}\|_0\\
{\rm s.t.}					&\quad\|{\bm U}_{k+1,\bigcdot}\|_0\leq z,\quad k\in\mathcal{T}.					
\end{aligned}
\end{equation}
\begin{remark}{\rm
In~\eqref{eq:U}, the requirements on $\zeta_{k,i}$ over controllers are represented by the row sparsity of $\bm U$, i.e., $\|{\bm U}_{k+1,\bigcdot}\|_0\leq z$, while the requirements over time are captured by the column sparsity, i.e., $\|{\bm U}_{\bigcdot,i}\|_0$.
}\end{remark}

To solve~\eqref{eq:U}, we first convexify the objective $\|{\bm U}_{\bigcdot,i}\|_0$ by replacing it with its $\ell_2$ surrogate. This relaxation can be further improved with reweighted $\ell_2$ minimization~\cite{chartrand2008iteratively} (introduced in Section~\ref{sec:reweighted}). Applying the $\ell_2$ relaxation leads to the following problem:
\begin{equation}\label{eq:l1_first}
\begin{aligned}
\min_{\bm U}	 		&\quad {\rm vec}(\bm U)'{\bm P}{\rm vec}(\bm U) + {\bm q}'{\rm vec}(\bm U)+\sum_{i=1}^N \alpha_i\|{\bm U}_{\bigcdot,i}\|^2_2\\
{\rm s.t.}					&\quad\|{\bm U}_{k+1,\bigcdot}\|_0\leq z,\quad k\in\mathcal{T}.					
\end{aligned}
\end{equation}
{\begin{remark}
	While $\ell_1$ relaxation is more commonly used and generally regarded as a more effective convex surrogate than plain $\ell_2$ relaxation, reweighted $\ell_2$ minimization has also been proposed as effective approximations of $\ell_p$ norms for $0<p<1$. As a result, it performs competitively with, and in some applications, outperforms (reweighted) $\ell_1$ minimization~\cite{daubechies2010iteratively}. Moreover, since $\ell_2$ relaxation enables a closed-form solution for the $\bm{U}$-minimization step given in~\eqref{eq:ls}, the computation cost is lower than $\ell_1$-based approaches. These advantages are also validated by the simulation results presented in Section~\ref{chap:sim}.
	\end{remark}}
{\begin{remark}
	 While the following proposed ADMM algorithm (see Section~\ref{sec:admm}) is efficient in practice, its convergence behavior can be sensitive to the choice of hyperparameters, and the iterates may exhibit oscillatory behavior~\cite{xu2016empirical}. By using the $\ell_2$ relaxation, we obtain a more robust algorithm with rigorous convergence guarantees (Theorem~\ref{thm:converge1} and Theorem~\ref{thm:converge2}).
\end{remark}}

Despite the usage of the $\ell_2$ surrogate, the constraint $\|{\bm U}_{k+1,\bigcdot}\|_0\leq z$ in~\eqref{eq:l1_first} remains challenging to handle. In the following, we address this using the ADMM technique~\cite{boyd2011distributed,shi2022cardinality}.

\subsection{ADMM}\label{sec:admm}
Introducing an auxiliary variable $\bm V$,~\eqref{eq:l1_first} can be written as
\begin{equation*}
\begin{aligned}
\min_{{\bm U},{\bm V}}	 		&\quad  {\rm vec}(\bm U)'{\bm P}{\rm vec}(\bm U) + {\bm q}'{\rm vec}(\bm U)+\sum_{i=1}^N \alpha_i\|{\bm U}_{\bigcdot,i}\|^2_2\\
							&\quad\quad +\sum_{k=0}^{T-1}\mathcal{I}\left({\bm V}_{k+1,\bigcdot}\right)\\
{\rm s.t.}					&\quad {\bm U}={\bm V},				
\end{aligned}
\end{equation*}
where $\mathcal{I}({\bm x})$ is an indicator function such that $\mathcal{I}({\bm x})=0$ if $\|{\bm x}\|_{0}\leq z$ and $\mathcal{I}({\bm x})=+\infty$ otherwise. Its augmented Lagrangian is then given by
\begin{equation*}
\begin{aligned}
\mathcal{L}({\bm U},{\bm V},{\bm \Lambda})	&={\rm vec}(\bm U)'{\bm P}{\rm vec}(\bm U) + {\bm q}'{\rm vec}(\bm U)\\
					&+\sum_{i=1}^N \alpha_i\|{\bm U}_{\bigcdot,i}\|^2_2+\sum_{k=0}^{T-1}\mathcal{I}\left({\bm V}_{k+1,\bigcdot}\right)\\
					&+{\rm Tr}[{\bm \Lambda}'({\bm U}-{\bm V})]+\frac{\rho}{2}\|{\bm U}-{\bm V}\|_F^2,
\end{aligned}
\end{equation*}
where ${\bm \Lambda}\in\mathbb{R}^{T\times N}$ is the Lagrangian multiplier and $\rho>0$ is the parameter for the regularization term.

The ADMM algorithm is then structured into the following three alternating steps:
\begin{equation*}
\begin{cases}
{\bm V}\text{\rm-Minimization}: {\bm V}^{t+1}=\arg\min_{\bm V}\mathcal{L}({\bm U}^t,{\bm V},{\bm \Lambda}^t);\\
{\bm U}\text{\rm-Minimization}: {\bm U}^{t+1}=\arg\min_{\bm U}\mathcal{L}({\bm U},{\bm V}^{t+1},{\bm \Lambda}^t);\\
{\bm \Lambda}\text{\rm-Update}: {\bm \Lambda}^{t+1}={\bm \Lambda}^t+\rho({\bm U}^{t+1}-{\bm V}^{t+1}),
\end{cases}
\end{equation*}
where $(\cdot)^t$ represents the $t$-th iteration of the algorithm.

Next, we derive closed-form solutions for each of the three ADMM update steps. It is worth noting that, instead of simply treating ADMM as a black-box solver, our proposed ADMM algorithm offers three advantages: (i) the updates admit closed-form expressions, (ii) the updates have physical interpretations, and (iii) convergence to a stationary point is guaranteed.

\subsubsection{Development of {\boldmath $V$}-Minimization} Denote ${\bm Y}^t={\bm U}^t+{{\bm \Lambda}^t}/{\rho}$. The following lemma gives the $\bm V$-minimization step.
\begin{lemma}\label{lemma:vmin}{\rm
The solution to ${\bm V}^{t+1}=\arg\min_{\bm V}\mathcal{L}({\bm U}^t,{\bm V},{\bm \Lambda}^t)$ is determined by
\begin{equation}\label{eq:v-min}
{\bm V}^{t+1}_{k+1,i}=
\begin{cases}
{\bm Y}_{k+1,i}^t,	&\text{if }\|{\bm Y}^t_{k+1,i}\|_2\geq \eta^t_{k+1};\\
0, 	&\text{otherwise},
\end{cases}
\end{equation}
where $\eta^t_{k+1}$ is the $z$-th largest value of $\|{\bm Y}^t_{k+1,i}\|_2$ for $i\in\mathcal{N}$.
}\end{lemma}
\begin{proof}
The $\bm V$-minimization step is equivalent to
\begin{equation*}
{\bm V}^{t+1}	=\arg\min_{\bm V}\left\{\sum_{k=0}^{T-1}\mathcal{I}\left({\bm V}_{k+1,\bigcdot}\right)+\frac{\rho}{2}\|{\bm Y}^t-{\bm V}\|_F^2\right\},
\end{equation*}
which can be again equivalently transformed into the following optimization problem:
\begin{equation}\label{eq:x-min}
\begin{aligned}
\min_{\bm V}	 		&\quad \|{\bm Y}^t-{\bm V}\|_F^2\\
{\rm s.t.}					&\quad \|{\bm V}_{k+1,\bigcdot}\|_0\leq z,\quad k\in\mathcal{T}.
\end{aligned}
\end{equation}
By considering each row of $\bm V$ separately,~\eqref{eq:x-min} can be decomposed into $T$ optimization problems. For $k\in\mathcal{T}$, we then have
\begin{equation}\label{eq:decompose}
\begin{aligned}
\min_{\bm V}	 		&\quad \|{\bm Y}_{k+1,\bigcdot}^t-{\bm V}_{k+1,\bigcdot}\|_2^2\\
{\rm s.t.}					&\quad \|{\bm V}_{k+1,\bigcdot}\|_0\leq z.
\end{aligned}
\end{equation}

Let $\Upsilon_{k+1}$ be the index set of nonzero elements in ${\bm V}_{k+1,\bigcdot}$, and $\Upsilon^c_{k+1}=\mathcal{N}/\Upsilon_{k+1}$ be that of zero elements. Obviously, we require $|\Upsilon_{k+1}|\leq z$. The objective of~\eqref{eq:decompose} is then given by
\begin{equation*}
\begin{aligned}
&\|{\bm Y}_{k+1,\bigcdot}^t-{\bm V}_{k+1,\bigcdot}\|_2^2	\\
=&\sum_{i\in\Upsilon_{k+1}}\|{\bm Y}_{k+1,i}^t-{\bm V}_{k+1,i}\|_2^2+\sum_{i\in\Upsilon^c_{k+1}}\|{\bm Y}_{k+1,i}^t-{\bm V}_{k+1,i}\|_2^2\\
									\leq&\sum_{i\in\Upsilon^c_{k+1}}\|{\bm Y}_{k+1,i}^t\|_2^2.
\end{aligned}
\end{equation*}
To minimize the objective function, $\Upsilon_{k+1}^c$ should include the $N-z$ smallest absolute values of $\|{\bm Y}_{k+1,i}^t\|_2$ for $i\in\mathcal{N}$. Thus, the solution is determined by~\eqref{eq:v-min}.
\end{proof}
\subsubsection{Development of {\boldmath $U$}-Minimization}
The $\bm U$-minimization step is equivalent to
\begin{equation}\label{eq:l1}
\begin{aligned}
{\bm U}^{t+1}&=\arg\min_{\bm U}{\rm vec}(\bm U)'{\bm P}{\rm vec}(\bm U) + {\bm q}'{\rm vec}(\bm U)\\
					&+\sum_{i=1}^N \alpha_i\|{\bm U}_{\bigcdot,i}\|^2_2 +{\rm Tr}\left[({\bm \Lambda}^t)'{\bm U}\right]+\frac{\rho}{2}\|{\bm U}-{\bm V}^{t+1}\|_F^2.
\end{aligned}
\end{equation}

Let $\bar{\bm u}^{t+1}_i$, ${\bm v}_i^{t+1}$ and $\bm\lambda_i^{t}$ denote the $i$-th column of ${\bm U}^{t+1}$, ${\bm V}^{t+1}$ and $\bm\Lambda^t$, respectively.  By considering each column $i$ separately, we decompose~\eqref{eq:l1} into $N$ problems, i.e.,
\begin{equation*}
\begin{aligned}
\bar{\bm u}^{t+1}_i&=\arg\min_{\bar{\bm u}_i}\bar{\bm u}_i'{\bm P}_i\bar{\bm u}_i + {\bm q}_i'\bar{\bm u}_i+\alpha_i\|\bar{\bm u}_i\|^2_2\\
					&\quad +({\bm \lambda}_i^t)'\bar{\bm u}_i+\frac{\rho}{2}\|\bar{\bm u}_i-{\bm v}_i^{t+1}\|_2^2\\
					&=\arg\min_{\bar{\bm u}_i}\frac{1}{2}\bar{\bm u}_i'{\bm \Phi}_i\bar{\bm u}_i-(\rho{\bm v}^{t+1}_i-{\bm q}_i-{\bm \lambda}_i^t)'\bar{\bm u}_i,
\end{aligned}
\end{equation*}
where $i\in\mathcal{N}$ and ${\bm \Phi}_i=2{\bm P}_i+(\rho+2\alpha_i){\bm I}$. Since ${\bm \Phi}_i\succ 0$, we obtain
\begin{equation}\label{eq:ls}
\bar{\bm u}^{t+1}_i	={\bm \Phi}_i^{-1}(\rho{\bm v}^{t+1}_i-{\bm q}_i-{\bm \lambda}_i^t).
\end{equation}

\subsubsection{Separation Principle} While ADMM is commonly used for non-convex optimization problems, it has a clear physical interpretation in our specific problem. Recall that our goal is to schedule the control data both over controllers and over time. The proposed ADMM algorithm effectively ``separates'' this co-design process into two iterative steps: one over controllers and the other over time. Specifically:
\begin{itemize}
	\item ${\bm V}$-Minimization:  Recall~\eqref{eq:x-min}. Starting from ${\bm U}^t$, which may not be feasible to the original problem~\eqref{eq:l1_first}, this step generates a feasible solution ${\bm V}^{t+1}$ while keeping the distance between ${\bm U}^t$ and ${\bm V}^{t+1}$, i.e., $\|\bm U^t-\bm V^{t+1}\|_F$, small. Thus, this step can be regarded as the {\it scheduling-over-controllers} step.
	\item ${\bm U}$-Minimization: Recall~\eqref{eq:l1}. This step yields a column-sparse-promoting ${\bm U}^{t+1}$ while maintaining a small distance between ${\bm U}^{t+1}$ and ${\bm V}^{t+1}$. Thus, this step can be viewed as the {\it scheduling-over-time} step.
\end{itemize}
Fig.~\ref{fig:admm} provides a visual illustration of the iterative steps of the ADMM algorithm.

\begin{figure}[!htbp]
	\centering
	\includegraphics[width=0.8\linewidth]{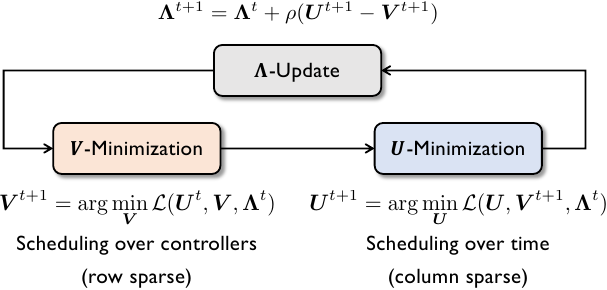}
	\caption{Illustration of the iterative steps in ADMM.}
	\label{fig:admm}
\end{figure}

\subsubsection{Theoretical Analysis} This part provides a theoretical analysis of the proposed ADMM algorithm.

Let $\bar{\omega}_i$ and $\underline{\omega}_i$ represent the largest and smallest eigenvalues of $\bm{P}_i+\alpha_i\bm{I}$, respectively. Define $\bar{\omega}\triangleq\max_i\bar{\omega}_i$ and $\underline{\omega}\triangleq\min_i\underline{\omega}_i$.
\begin{lemma}\label{lemma:iteration}{\rm
If $\rho>{4\bar{\omega}^2}/{\underline{\omega}}$, then for each iteration $t$, there exists a constant $\kappa>0$ such that
\begin{equation*}
\mathcal{L}({\bm U}^{t+1},{\bm V}^{t+1},{\bm \Lambda}^{t+1})-\mathcal{L}({\bm U}^t,{\bm V}^t,{\bm \Lambda}^t)\leq-\kappa\|{\bm U}^{t+1}-{\bm U}^t\|_F^2.
\end{equation*}
}\end{lemma}
\begin{proof}
See Appendix~\ref{apx:iteration}.
\end{proof}
\begin{lemma}\label{lemma:lbound}{\rm
If $\rho>2{\bar{\omega}}$, then the sequence $\{\mathcal{L}({\bm U}^t,{\bm V}^t,{\bm \Lambda}^t)\}_{t=0}^\infty$ is lower bounded.
}\end{lemma}
\begin{proof}
See Appendix~\ref{apx:lbound}.
\end{proof}
\begin{thm}\label{thm:converge1}{\rm
If $\rho>\max\{{4\bar{\omega}}^2/{\underline{\omega}},2{\bar{\omega}}\}$, then the sequence $\{\mathcal{L}({\bm U}^t,{\bm V}^t,{\bm \Lambda}^t)\}_{t=0}^\infty$ converges as $t\to\infty$.
}\end{thm}
\begin{proof}
By Lemma~\ref{lemma:iteration} and Lemma~\ref{lemma:lbound}, $\mathcal{L}({\bm U}^t,{\bm V}^t,{\bm \Lambda}^t)$ is monotonically decreasing and lower bounded, and therefore convergent as $t\to\infty$.
\end{proof}
\begin{thm}\label{thm:converge2}{\rm
If $\rho>\max\{{4\bar{\omega}}^2/{\underline{\omega}},2{\bar{\omega}}\}$, then we obtain $\|\bm U^{t+1}-\bm U^t\|_F^2\to0$, $\|\bm V^{t+1}-\bm V^t\|_F^2\to0$ and $\|\bm \Lambda^{t+1}-\bm \Lambda^t\|_F^2\to0$ as $t\to\infty$.
}\end{thm}
\begin{proof}
See Appendix~\ref{apx:converge2}.
\end{proof}
Theorems~\ref{thm:converge1} and~\ref{thm:converge2} establish the convergence of our proposed ADMM algorithm. Although ADMM is widely applied to non-convex optimization problems, deriving rigorous convergence guarantees for the general non-convex setting remains challenging~\cite{magnusson2015convergence}. {Consequently, such convergence results are extremely limited in the existing literature.}

Having secured convergence, we now examine the ``quality'' of the solution. { Since~\eqref{eq:l1_first} is non-convex, certifying global optimality is generally intractable. Instead, it is standard to characterize solutions through first-order necessary optimality conditions, i.e., the Karush-Kuhn-Tucker (KKT) conditions.} However, due to the non-smooth cardinality constraints in~\eqref{eq:l1_first}, the standard KKT framework is not directly applicable. To overcome this, we adopt the notion of $L$-stationarity, {which generalizes first-order stationarity to cardinality-constrained problems.}
\begin{definition}[$L$-stationarity~\cite{beck2013sparsity}]\label{def:L}{\rm
For the problem $\{\min f(\bm x):\bm x\in\mathcal{X}\}$ with $f(\cdot):\mathbb{R}^n\to\mathbb{R}$ continuously differentiable, a point $\bm x$ is an $L$-stationary point if
\begin{equation*}
\bm x\in\mathcal{P}_{\mathcal{X}}\left(\bm x-\frac{1}{L}\nabla f(\bm x)\right),
\end{equation*}
where $\nabla f(\bm x)$ is the gradient of $f(\bm x)$ and $\mathcal{P}_{\mathcal{X}}(\cdot)$ is the orthogonal projection onto $\mathcal{X}$.
}\end{definition}
\begin{remark}{\rm
Definition~\ref{def:L} extends the standard notion of a stationary point. In the special case where $\mathcal{X}$ is closed and convex, $\bm x$ is a stationary point if and only if the equality $\bm x=\mathcal{P}_{\mathcal{X}}(\bm x-\frac{1}{L}\nabla f(\bm x))$ holds for any $L>0$~\cite{bertsekas1997nonlinear}.
}\end{remark}
\begin{prop}\label{prop:necessary}{\rm
If $\bm U^\star$ is an optimal solution to~\eqref{eq:l1_first}, then ${\rm vec}(\bm U^\star)$ is an $L$-stationary point to~\eqref{eq:l1_first} for any $L>2\bar{\omega}$.
}\end{prop}
\begin{proof}
The proof follows directly from~\cite[Theorem 2.2]{beck2013sparsity} and is therefore omitted.
\end{proof}
{\begin{remark}
	Proposition~\ref{prop:necessary} shows that $L$-stationarity with $L>2\bar{\omega}$ is a necessary condition for optimality. It is important to note that $\bar{\omega}$ is not a consequence of our algorithm design. Instead, it is introduced as part of the definition of $L$-stationarity in~[31], and the requirement $L>2\bar{\omega}$ is an intrinsic characteristic of the problem formulation. 
\end{remark}}
In the following, we use the superscript $(\cdot)^{\infty}$ to denote the corresponding limit values generated by ADMM.

\begin{thm}\label{thm:stationary}{\rm
{For any $L\geq\rho$}, ${\rm vec}(\bm U^\infty)$ is a $L$-stationary solution to~\eqref{eq:l1_first}.
}\end{thm}
\begin{proof}
See Appendix~\ref{apx:stationary}.
\end{proof}
{ Since $\rho>\max\{{4\bar{\omega}}^2/{\underline{\omega}},2{\bar{\omega}}\}$, Theorem~\ref{thm:stationary} ensures that the limit point by ADMM, i.e., ${\rm vec}(\bm U^\infty)$, is at least an $L$-stationary point for any $L>\max\{{4\bar{\omega}}^2/{\underline{\omega}},2{\bar{\omega}}\}$. In other words, ${\rm vec}(\bm U^\infty)$ satisfies this necessary optimality condition.}
\subsection{Reweighted Minimization}\label{sec:reweighted}
The $\ell_2$-relaxation in~\eqref{eq:l1_first} penalizes larger numbers more heavily than the $\ell_0$-norm does. To better approximate the behavior of $\ell_0$ norm, we adopt the reweighted $\ell_2$ minimization approach~\cite{chartrand2008iteratively}. It increases the weights for small elements in the norm, i.e., replacing $\|\bar{\bm u}_i\|_2^2$ with $\bar{\bm u}_i'\bm{W}_i\bar{\bm u}_i$, where\footnote{The variable $\epsilon$ is a small positive constant introduced to ensure the inversion is well-defined.}
\begin{equation*}
{\bm W}_{i}={\rm diag}[{\bm 1}\oslash (\bar{\bm u}_i\odot \bar{\bm u}_i+\epsilon)].
\end{equation*}
By iteratively updating these weights, the reweighted method reduces the excessive penalization of larger numbers. Note that in this case,~\eqref{eq:ls} is slightly modified into
\begin{equation}\label{eq:reweight}
\bar{\bm u}^{t+1}_i	=(2\bm P_i+2\alpha_i\bm W_i+\rho^t \bm I)^{-1}(\rho^t{\bm v}^{t+1}_i-{\bm q}_i-{\bm \lambda}_i^t).
\end{equation}

\subsection{Strategy Refinement}
In~\eqref{eq:u}, we consolidate the design of $\zeta_{k,i}$ and $\bm u_{k,i}$ into a single variable, i.e., $\tilde{\bm u}_{k,i}=\zeta_{k,i}\bm u_{k,i}$. However, due to the $\ell_2$ relaxation in~\eqref{eq:l1_first}, this formulation tends to over-penalize nonzero entries in $\bm U$, driving them unnecessarily small {rather than achieving true sparsity. Consequently, the resulting $\bm U$ is not optimal for the underlying schedule and thus can be further improved.}

{To ``unpenalize'' and restore these overly suppressed elements, we introduce a strategy refinement step.} Specifically, we use the obtained matrix $\bm U$ to inform our scheduling variable $\zeta_{k,i}$, and then refine and recompute the control signals $\bm u_{k,i}$ based on that schedule.

Firstly, we extract $\zeta_{k,i}$ from $\bm U$. Recall
\begin{equation*}
\bm U=
\begin{bmatrix}
\tilde{\bm u}_{0,1} 	& \tilde{\bm u}_{0,2} 	&\cdots	&\tilde{\bm u}_{0,N}\\
\tilde{\bm u}_{1,1} 	& \tilde{\bm u}_{1,2} 	&\cdots	&\tilde{\bm u}_{1,N}\\
\vdots 			&\vdots			&\ddots 	&\vdots\\
\tilde{\bm u}_{T-1,1} 	& \tilde{\bm u}_{T-1,2} 	&\cdots	&\tilde{\bm u}_{T-1,N}
\end{bmatrix}.
\end{equation*}
The scheduling strategy $\zeta_{k,i}$ is then computed as
\begin{equation}\label{eq:zeta_compute}
\zeta_{k,i}=
\begin{cases}
0 &\text{if $\bm{\|\tilde{u}}_{k,i}\|_2=0$};\\
1 & \text{otherwise},
\end{cases}
\end{equation}
Then the control signal $\bm{u}_{k,i}$ can be directly computed as described in the following lemma.

\begin{lemma}[Proposition 3.1~\cite{shi2012finite}]\label{lemma:ucompute}{\rm
For each system $i$, given a schedule $\{\zeta_{k,i}\}_{k=0}^{T-1}$, the control signal $\bm u_{k,i}$ that minimizes $J_i$ is expressed as
\begin{equation}\label{eq:u_compute}
{\bm u}_{k,i}=-[\bm B_i'\bm S_{k+1,i}\bm B_i+\bm R_i]^{-1}\bm B_i'\bm S_{k+1,i}\bm A_i\bm{x}_{k,i},
\end{equation}
where $\bm S_{k,i}$ is computed recursively as
\begin{equation*}
\begin{aligned}
\bm S_{k,i}&=\bm A_i'\bm S_{k+1,i}\bm A_i+\bm Q_i\\
&\quad-\zeta_{k,i}\bm A'_i\bm S_{k+1,i}\bm B_i[\bm B_i'\bm S_{k+1,i}\bm B_i+\bm R_i]^{-1}\bm B_i'\bm S_{k+1,i}\bm A_i,
\end{aligned}
\end{equation*}
with initial condition $\bm S_{T,i}=\bm Q_i$.
}\end{lemma}

{For any given schedule $\{\zeta_{k,i}\}_{k=0}^{T-1}$, Lemma~\ref{lemma:ucompute} provides the optimal control signal $\bm u_{k,i}$ in a simple closed form. Consequently, the refinement step is guaranteed to yield a lower objective value for Problem~\eqref{eq:U}, while incurring only negligible computational overhead.} A summary of the proposed $\ell_2$ reweighted ADMM algorithm with strategy refinement is provided in Algorithm~\ref{alg:admm}.\footnote{The computational complexity is primarily driven by the $\bm V$- and $\bm U$-minimization steps, whose dominant operations are matrix inversion and sorting, respectively. Consequently, the overall time complexity is $O(\sum_{i=1}^Nm_i^3T^3+TN\log N)$, while the memory complexity is $O((\max_im_i)^2T^2 + N)$.}

\begin{algorithm}
	\caption{The $\ell_2$ reweighted ADMM algorithm.}
	\DontPrintSemicolon
	\label{alg:admm}
		\KwData{collision-free requirement $z$, penalty weight $\alpha_i$, regularization variable $\rho$}
		\KwResult{control data $\bm u_{k,i}$ and its schedules $\zeta_{k,i}$}
			initialize ${\bm \Lambda}^0={\bm 0}_{m_iT\times N}$\;
			initialize ${\bm U}^0$ by solving~\eqref{eq:l1_first} without constraints\;
			\Repeat(\tcc*[h]{reweighted minimization}){\rm convergence or maximum iterations}{
				compute $\bm W_i={\rm diag}[{\bm 1}\oslash (\bar{\bm u}_i^0\odot\bar{\bm u}_i^0 +\epsilon)]$\;
				$t\gets0$\;
				\Repeat(\tcc*[h]{ADMM}){\rm convergence or maximum iterations}{
					a) {$\bm V$-minimization}: compute ${\bm V}^{t+1}$ via~\eqref{eq:v-min}\;
					b) {$\bm U$-minimization}: obtain ${\bm U}^{t+1}$ via~\eqref{eq:reweight}\;
					c) {${\bm \Lambda}$-update}: ${\bm \Lambda}^{t+1}={\bm \Lambda}^t+\rho^t({\bm U}^{t+1}-{\bm V}^{t+1})$\;
					$t\gets t+1$
				}
				$\bm\Lambda^0\gets\bm\Lambda^t$\;
				compute ${\bm V}^{t+1}$ via~\eqref{eq:v-min} and ${\bm U}^{0}\gets{\bm V}^{t+1}$
			}
			compute $\zeta_{k,i}$ via~\eqref{eq:zeta_compute} with ${\bm U}\gets{\bm U}^{0}$\;
			compute $\bm u_{k,i}$ via~\eqref{eq:u_compute}
\end{algorithm}
\section{Simulation}\label{chap:sim}
{
In this section, we evaluate the performance of our proposed Algorithm~\ref{alg:admm} through two sets of simulations: (i) a case study comparing its performance with that of $\ell_1$-based approaches, and (ii) an extensive evaluation on a large collection of randomly generated spatially distributed systems.}

It is important to note that the choice of $\rho$ influences the convergence speed of the ADMM algorithm~\cite{magnusson2015convergence}. {Throughout the theoretical development in Section~\ref{sec:admm}, we consider a fixed $\rho$ value, and establish convergence guarantees under the condition $\rho>\max\{{4\bar{\omega}}^2/{\underline{\omega}},2\bar{\omega}\}$. In this part, however, we adopt the following adaptive update rule}\footnote{Since $\theta > 1$, there exists a finite iteration index $t_0$ such that $\rho^t=\rho_{\max}$ for all $t\geq t_0$. Hence, after $t_0$, the adaptive algorithm coincides with the fixed-$\rho$ scheme analyzed previously in Theorems~\ref{thm:converge1} and \ref{thm:converge2}. Nevertheless, a rigorous extension would require a separate examination of the transient phase $t<t_0$, which is beyond the scope of the present work.}: 
\begin{equation*}
\rho^{t+1}=\min\{\theta\rho^t,\rho_{\max}\},
\end{equation*}
where $\theta>1$ is a scaling parameter, and $\rho_{\max}>\max\{{4\bar{\omega}}^2/{\underline{\omega}},2{\bar{\omega}}\}$ is a prescribed upper bound. This heuristic allows aggressive updates during the early iterations, while gradually increasing $\rho$ to enhance convergence as the algorithm progresses. Such adaptive strategies have been reported to perform well in practice~\cite{magnusson2015convergence,peng2015proximal}.

Throughout the simulations, we use $\rho^0=0.004$, $\rho_{\max}=40$, and $\theta=1.2$. Moreover, we set $N=4$, $z=3$, $\bm{Q}_i=\bm{I}$ and $\bm{R}_i=\bm{I}$ for $i=1,\dots,N$. For notational simplicity, we denote $\bm{\alpha}=[\alpha_1,\dots,\alpha_N]$.

\subsection{Case Study}
Consider the scenario where the following $N=4$ systems operate over a shared network:
\begin{equation*}
\begin{aligned}
\bm A_1&=\begin{bmatrix}0.8930&-0.00062\\2.7738&0.8864\end{bmatrix},
\bm A_2=\begin{bmatrix}0.8895&-0.00294\\0.9447&0.9968\end{bmatrix},\\
\bm A_3&=\begin{bmatrix}0.9654&-0.00182\\-0.6759&0.9433\end{bmatrix},
\bm A_4=\begin{bmatrix}0.8227&-0.00168\\6.1233&0.9367\end{bmatrix},\\
\bm B_1&=\begin{bmatrix}-0.000034&0.0986\\-0.000157&0.1045\end{bmatrix}',
\bm B_2=\begin{bmatrix}-0.000097&0.1016\\ -0.000092&0.1014\end{bmatrix}',\\
\bm B_3&=\begin{bmatrix}-0.000096&0.1016\\ -0.000127&0.1031\end{bmatrix}',
\bm B_4=\begin{bmatrix}-0.000095&0.1015\\-0.000063&0.1000\end{bmatrix}'.
\end{aligned}
\end{equation*}
The initial state of each system is a 2-dimensional random vector, with each element sampled independently from a uniform distribution over the interval $(0,1)$. The horizon length is set to $T=30$.

{Fig.~\ref{fig:compare} shows the resulting total number of transmissions (NoT), i.e., $\sum_{k=0}^{T-1}\sum_{i=1}^N\zeta_{k,i}$, and the corresponding cost value $J$ for different choices of $\bm{\alpha}$. We observe that both reweighted approaches outperform plain $\ell_1$ minimization, yielding a more favorable trade-off between actuator effort and performance. Secondly, although reweighted $\ell_1$ and reweighted $\ell_2$ achieve comparable sparsity-recovery performance, our proposed $\ell_2$-based algorithm is more computationally efficient. As reported in Table~\ref{tb:l1}, it requires approximately half the number of iterations and exhibits a significantly lower per-iteration computation time. These results indicate that, besides offering theoretical guarantees, the proposed formulation is also effective in practice.

Fig.~\ref{fig:compare} also illustrates the effect of the refinement step. The amount of improvement depends strongly on the quality of the sparse approximation obtained previously. When sparsity is promoted less accurately (plain $\ell_1$ minimization), the refinement procedure can substantially reduce the resulting cost $J$. In contrast, for the two reweighted methods, since the weighting scheme already partially mitigates the shrinkage effect, refinement yields a smaller improvement. Nevertheless, refinement consistently lowers the cost in all three approaches, confirming its effectiveness as a post-processing step.}

\begin{figure}[!htbp]
	\centering
	\includegraphics[width=\linewidth]{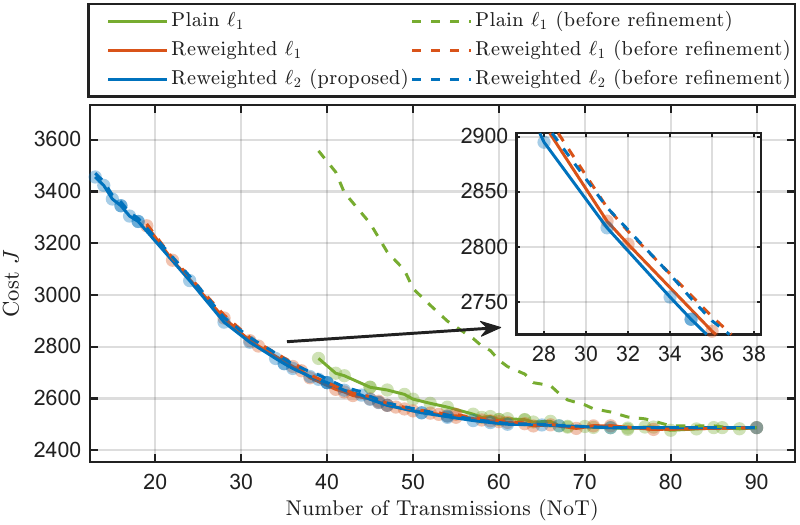}
	\caption{Trade-offs between NoT and $J$. The curves are generated by varying $\bm{\alpha}$ from $\bm{0}_4$ to $\bm{10}_4$.}
	\label{fig:compare}
\end{figure}
    
\begin{table}[!htbp]
	\centering
	
	\caption{Average computational burden of the reweighted $\ell_1$ and reweighted $\ell_2$ schemes over 36 simulation runs.}
	\label{tb:l1}
	\begin{tabular}{lcc}
	\hline
                 		& Reweighted $\ell_1$  & Reweighted $\ell_2$  \\ \hline
	Number of Iterations     	& 97.1388  & 49.7222  \\
	Time per Iteration (s) 	& 15.4248  & 0.0705  \\ \hline
\end{tabular}
\end{table}
{
\subsection{Spatially Distributed Systems}
\begin{table*}[!htbp]
	\caption{Average relative NoT and relative control cost $J$ for randomly generated spatially distributed systems.}
	\label{tb:random}
	\centering
	\begin{tabular}{|c|cccc|cccc|cccc|}
		\hline
		\multicolumn{1}{|c|}{} & \multicolumn{4}{c|}{Stable} & \multicolumn{4}{c|}{Unstable} & \multicolumn{4}{c|}{Mixed} \\ \hline
		$\bm\alpha$    & $\bm{0}_4$       &  $\bm{1}_4$      & $\bm{2}_4$ &  $\bm{5}_4$       &  $\bm{0}_4$        &  $\bm{1}_4$   & $\bm{2}_4$     &  $\bm{5}_4$       &  $\bm{0}_4$       &  $\bm{0.1}_4$     & $\bm{0.2}_4$  &  $\bm{0.5}_4$      \\ \hline
		Relative NoT &0.75        &  0.62       &   0.48      &       0.31             &  0.75       &  0.30       &  0.24       & 0.18    & 0.75 &0.60 &0.58&0.57\\ \hline
		Relative Cost $J$  &  1.08       & 1.09        &1.10         & 1.18         &  2.07        & 2.12        & 2.14        &2.15         &1.11  &1.12 &1.17     &1.20\\ \hline
	\end{tabular}
\end{table*}
We evaluate the proposed algorithm on a class of randomly generated spatially distributed networks inspired by~\cite{motee2008optimal}.

Consider $N=4$ systems, each consisting of two spatially distributed nodes, whose locations are generated randomly and independently within a $10\times10$ square region. For system $i$, the dynamics are
\begin{equation*}
	[\dot{\bm x}_i]_j = -[{\bm x}_i]_j + \sum_{j\neq s}e^{-d(j,s)} [{\bm x}_i]_s + \bm{u}_i,
\end{equation*}
where $[\cdot]_j$ is the $j$-th element of a vector, $d(j,s)$ is determined by the Euclidean distance between nodes $j$ and $s$. The compact form of the system is then given by
\begin{equation*}
	\dot{\bm x}_i = \bm{A}_{i}^c \bm{x}_i + \bm{u}_i, 
\end{equation*}
where $\bm{A}_{i}^c$ is uniquely determined by the node locations.

Considering a unit sampling time, the discrete system is
\begin{equation*}
\bm{x}_{k+1,i} = \bm{A}_i\bm{x}_{k,i} + \bm{B}_i \bm{u}_{k,i},
\end{equation*}
with $\bm{A}_i = e^{\bm{A}_{i}^c}$ and $\bm{B}_i=e^{\bm{A}_{i}^c}\int^1_0e^{-\bm{A}_{i}^c\tau}d\tau$. The initial state of each system is selected independently and uniformly from $(0,5)$. The horizon length is set to $T=10$.

It is easy to verify that the above generated systems are unstable. To evaluate our proposed algorithm under different dynamical characteristics, we generated 150 independent system realizations, which can be partitioned into the following three categories:
\begin{itemize}
	\item {\bf Unstable}: 50 realizations in which all four systems are generated according to the procedure described above;
	\item {\bf Stable}: 50 realizations in which the $\bm A_i$ matrices are first generated according to the above procedure, and subsequently scaled using their spectral radius so that all four systems are stable;
	\item {\bf Mixed}: 50 realizations consisting of two stable and two unstable systems, constructed from the same generation procedure.
\end{itemize}
The average relative NoT and cost $J$ are provided in Table~\ref{tb:random}, where the values are normalized by a baseline policy in which every controller transmits at every time step.\footnote{Since at most three out of four controllers can transmit simultaneously, the maximum achievable relative NoT under the collision-free constraint is $3/4 = 0.75$, which corresponds to the case where $\bm\alpha=\bm0_4$.}

The results demonstrate that our framework effectively balances actuator effort and performance across all system classes. In particular, increasing $\bm\alpha$ consistently reduces NoT while maintaining an acceptable increase in $J$, which indicates that our proposed framework is valid for various system dynamics. 

An interesting phenomenon is observed in the mixed scenario. Since unstable systems are significantly more sensitive to transmission absence, the proposed algorithm tends to allocate more communication resources to them. This behavior is reflected in the relative NoT, which remains close to $0.5$, suggesting a policy in which the unstable systems are scheduled frequently and the stable ones only intermittently.}

\section{Conclusion}
In this paper, we tackled the problem of control data scheduling in networked control systems with shared communication channels between remote controllers and actuators. We introduced a novel scheduling mechanism that jointly coordinates control data transmissions across multiple controllers and time steps and developed an ADMM-based algorithm for it. The algorithm is computationally efficient and offers a clear physical interpretation. Moreover, we prove that it is convergent and the returned solution is at least an $L$-stationary point of the relaxed problem, which is a necessary condition for optimality. Simulation results demonstrated that the proposed approach effectively reduces the communication burden and actuation effort while incurring a slight degradation in system performance. Future work may focus on analyzing the controllability/stability under the scheduling strategy.

\appendices
\section{Proof of Lemma~\ref{lemma:iteration}}\label{apx:iteration}
For the function $\mathcal{L}({\bm U}^t,{\bm V}^t,{\bm \Lambda}^t)$, we define the difference $\Delta$ between each iteration $t$ as 
\begin{equation}\label{eq:delta}
\begin{aligned}
\Delta	&\triangleq\mathcal{L}({\bm U}^{t+1},{\bm V}^{t+1},{\bm \Lambda}^{t+1})-\mathcal{L}({\bm U}^{t},{\bm V}^{t},{\bm \Lambda}^{t}).
\end{aligned}
\end{equation}

Since ${\bm V}^{t+1}=\arg\min_{\bm V}\mathcal{L}({\bm U}^t,{\bm V},{\bm \Lambda}^t)$, it follows that
\begin{equation}\label{eq:L1}
\mathcal{L}({\bm U}^{t},{\bm V}^{t+1},{\bm \Lambda}^{t})-\mathcal{L}({\bm U}^{t},{\bm V}^{t},{\bm \Lambda}^{t})\leq 0.
\end{equation} 

Additionally, since $\mathcal{L}({\bm U},{\bm V},{\bm \Lambda})$ is strongly convex in $\bm U$ and ${\bm U}^{t+1}=\arg\min_{\bm U}\mathcal{L}({\bm U},{\bm V}^{t+1},{\bm \Lambda}^t)$, we have
\begin{equation}\label{eq:L2}
\begin{aligned}
&\mathcal{L}({\bm U}^{t+1},{\bm V}^{t+1},{\bm \Lambda}^{t}) - \mathcal{L}({\bm U}^{t},{\bm V}^{t+1},{\bm \Lambda}^{t})\\
&\quad\leq-{\underline{\omega}}\|\bm{U}^{t+1}-\bm{U}^{t}\|_F^2.
\end{aligned}
\end{equation}

Finally, define $\mathcal{F}_i({\bm u})\triangleq\bm u'{\bm P}_i\bm u + {\bm q}_i'\bm u+\alpha_i\|{\bm u}\|^2_2$. With the optimality of $\bar{\bm u}^{t+1}_i$ in~\eqref{eq:ls}, we have
\begin{equation*}
\rho(\bm{u}_i^{t+1}-\bm{v}_i^{t+1})=-2(\bm{P}_i+\alpha_i\bm{I})\bm{u}_i^{t+1}-\bm{q}_i-\bm{\lambda}_i^t.
\end{equation*}
Combining it with the ${\bm \Lambda}$-update, i.e.,  ${\bm \Lambda}^{t+1}={\bm \Lambda}^t+\rho({\bm U}^{t+1}-{\bm V}^{t+1})$, we obtain
\begin{equation}\label{eq:important}
\bm\lambda_i^t=-\nabla\mathcal{F}_i(\bar{\bm u}_i^t), 
\end{equation}
where $\nabla \mathcal{F}_i({\bm u})=2(\bm P_i+\alpha_i \bm I)\bm u+\bm q_i$ is the gradient of $\mathcal{F}_i({\bm u})$.

Then, from the ${\bm \Lambda}$-update and~\eqref{eq:important}, we derive the following:
\begin{equation}\label{eq:L3}
\begin{aligned}
&~\mathcal{L}({\bm U}^{t+1},{\bm V}^{t+1},{\bm \Lambda}^{t+1}) - \mathcal{L}({\bm U}^{t+1},{\bm V}^{t+1},{\bm \Lambda}^{t})\\
=&~{\rm Tr}[({\bm \Lambda}^{t+1}-{\bm \Lambda}^{t})'({\bm U}^{t+1}-{\bm V}^{t+1})]\\
=&~\frac{1}{\rho}\|{\bm \Lambda}^{t+1}-{\bm \Lambda}^{t}\|_F^2=\frac{1}{\rho}\sum_{i=1}^N\|{\bm \lambda}_i^{t+1}-{\bm \lambda}_i^{t}\|_2^2\\
=&~\frac{4}{\rho}\sum_{i=1}^N\|(\bm{P}_i+\alpha_i\bm{I})(\bm{u}_i^{t+1}-\bm{u}_i^t)\|_2^2\leq\frac{4}{\rho}\bar{\omega}^2\|\bm{U}^{t+1}-\bm{U}^t\|_F^2.
\end{aligned}
\end{equation}

Substituting~\eqref{eq:L1}\eqref{eq:L2}, and \eqref{eq:L3} into~\eqref{eq:delta} and defining $\kappa \triangleq \underline{\omega} - 4\bar{\omega}^2/\rho$, we obtain
\begin{equation*}
\Delta\leq-\left({\underline{\omega}}-\frac{4}{\rho}\bar{\omega}^2\right)\|\bm{U}^{t+1}-\bm{U}^t\|_F^2=-\kappa\|\bm{U}^{t+1}-\bm{U}^t\|_F^2.
\end{equation*}
Since $\rho>{4\bar{\omega}^2}/{\underline{\omega}}$, we have $\kappa>0$, which completes the proof.
\section{Proof of Lemma~\ref{lemma:lbound}}\label{apx:lbound}
\begin{equation}\label{eq:q2p}
\mathcal{F}_i(\bm v)-\mathcal{F}_i({\bm u})\leq \nabla \mathcal{F}_i({\bm u})'(\bm v-{\bm u})+{\bar{\omega}_i}\|{\bm v}-{\bm u}\|_2^2.
\end{equation}

Substituting~\eqref{eq:important} into~\eqref{eq:q2p} and considering $\bm v_i^t$ and $\bar{\bm u}_i^t$, we obtain
\begin{equation}\label{eq:lp}
\mathcal{F}_i(\bm v_i^t)-{\bar{\omega}_i}\|{\bm v}_i^t-\bar{\bm u}_i^t\|_2^2\leq \mathcal{F}_i(\bar{\bm u}_i^t) +(\bm\lambda_i^t)'(\bar{\bm u}_i^t-{\bm v}_i^t).
\end{equation}

Therefore,
\begin{equation*}
\begin{aligned}
\mathcal{L}({\bm U}^t,{\bm V}^t,{\bm \Lambda}^t)	&=\sum_{i=1}^N\left[\mathcal{F}_i(\bar{\bm u}_i^t)+({\bm \lambda}_i^t)'(\bar{\bm u}_i^t-\bm v_i^t)\right]\\
										&\quad+\frac{\rho}{2}\sum_{i=1}^N\|\bar{\bm u}^t_i-{\bm v}_i^{t}\|_2^2+\sum_{k=0}^{T-1}\mathcal{I}\left({\bm V}_{k+1,\bigcdot}\right)\\
										&\geq \sum_{i=1}^N\left\{\mathcal{F}_i(\bm v_i^t)+(\frac{\rho}{2}-{\bar{\omega}_i})\|\bar{\bm u}_i^t-{\bm v}_i^t\|_2^2\right\}\\
										&\geq  \sum_{i=1}^N\mathcal{F}_i(\bm v_i^t),
\end{aligned}
\end{equation*}
where the first inequality follows from~\eqref{eq:lp} and the fact that $\sum_{k=0}^{T-1}\mathcal{I}\left({\bm V}_{k+1,\bigcdot}\right)\geq0$, and the second inequality holds because $\rho>2{\bar{\omega}}$. Furthermore, $\mathcal{F}_i(\bm u)$ is lower bounded since it is quadratic with $\bm P_i+\alpha_i\bm I\succ0$. Therefore, we conclude that $\mathcal{L}({\bm U}^t,{\bm V}^t,{\bm \Lambda}^t)$ is lower bounded, thus completing the proof.
\section{Proof of Theorem~\ref{thm:converge2}}\label{apx:converge2}
From Theorem~\ref{thm:converge1}, we know $\{\mathcal{L}({\bm U}^t,{\bm V}^t,{\bm \Lambda}^t)\}_{t=0}^\infty$ converges, i.e., $\mathcal{L}({\bm U}^{t+1},{\bm V}^{t+1},{\bm \Lambda}^{t+1})-\mathcal{L}({\bm U}^t,{\bm V}^t,{\bm \Lambda}^t)\to0$ as $t\to\infty$. Therefore, Lemma~\ref{lemma:iteration} indicates that $\|{\bm U}^{t+1}-{\bm U}^t\|_F^2\to0$.

From~\eqref{eq:L3}, we obtain $\|{\bm \Lambda}^{t+1}-{\bm \Lambda}^{t}\|_F^2\leq4\bar{\omega}\|\bm{U}^{t+1}-\bm{U}^t\|_F^2\to0$ as $t\to\infty$.

To analyze the convergence of $\|{\bm V}^{t+1}-{\bm V}^t\|_F^2$, we use the $\bm \Lambda$-update. Specifically, we have
\begin{equation*}
\begin{aligned}
	&~\|{\bm V}^{t+1}-{\bm V}^t\|_F^2\\
=	&~\|\bm U^{t+1}-\bm U^t-\frac{1}{\rho}(\bm\Lambda^{t+1}-\bm\Lambda^{t})+\frac{1}{\rho}(\bm\Lambda^{t}-\bm\Lambda^{t-1})\|_F^2\\
\leq	&\left[\|\bm U^{t+1}-\bm U^t\|_F+\frac{1}{\rho}\|\bm\Lambda^{t+1}-\bm\Lambda^{t}\|_F+\frac{1}{\rho}\|\bm\Lambda^{t}-\bm\Lambda^{t-1}\|_F\right]^2\\
\leq 	&~3\|\bm U^{t+1}-\bm U^t\|_F^2+\frac{3}{\rho^2}\|\bm\Lambda^{t+1}-\bm\Lambda^{t}\|_F^2+\frac{3}{\rho^2}\|\bm\Lambda^{t}-\bm\Lambda^{t-1}\|_F^2.
\end{aligned}
\end{equation*}
As $t \to \infty$, each term on the right-hand side converges to zero, $\|{\bm V}^{t+1}-{\bm V}^t\|_F^2 \to 0$ as well.

\section{Proof of Theorem~\ref{thm:stationary}}\label{apx:stationary}
We assume that the limit point $({\bm U}^{\infty},{\bm V}^{\infty},{\bm \Lambda}^{\infty})$ is unique. If multiple limit points exist, the subsequent derivation remains valid by considering an appropriate convergent subsequence.

To prove Theorem~\ref{thm:stationary}, it suffices to verify the following conditions~\cite[Lemma 2.2]{beck2013sparsity}:
\begin{enumerate}[(I)]
\item For $k\in\mathcal{T}$, $\|{\bm U}^\infty_{k+1,\bigcdot}\|_0\leq z$;
\item For $k\in\mathcal{T}$ and $i\in\mathcal{N}$, 
\begin{equation*}
\|\nabla_{k+1} \mathcal{F}_i({\bm{\bar{u}}_i^\infty})\|_2
\begin{cases}
\leq L\eta_{k+1}	&{\rm if~} {\bm U}_{k+1,i}^\infty=0;\\
=0				&{\rm otherwise},
\end{cases}
\end{equation*}
where $\mathcal{F}_i(\cdot)$ is defined in Appendix~\ref{apx:iteration}, $\nabla_{k+1} \mathcal{F}_i({\cdot})$ denotes the $k+1$-th component of its gradient and $\eta_{k+1}$ is the $z$-th largest value of $\|\bm U^\infty_{k+1,i}\|_2$ for $i\in\mathcal{N}$.
\end{enumerate}

We begin by recalling Theorem~\ref{thm:converge2} and the ADMM update steps. The convergence of $\bm\Lambda^{t+1}-\bm\Lambda^t$ implies
\begin{equation*}
{\bm U}^{\infty}-{\bm V}^{\infty}=\lim_{t\to\infty}\left({\bm U}^{t+1}-{\bm V}^{t+1}\right)=\lim_{t\to\infty}\frac{{\bm \Lambda}^{t+1}-{\bm \Lambda}^t}{\rho}=0.
\end{equation*}
Thus, $\|{\bm U}^\infty_{k+1,\bigcdot}\|_0=\|{\bm V}^\infty_{k+1,\bigcdot}\|_0\leq z$, which proves (I).

We now verify (II) by analyzing the ${\bm V}$-minimization step given in~\eqref{eq:v-min}. Consider the two cases:

\noindent \underline{\it Case 1:} $\|{\bm Y}^t_{k+1,i}\|_2\geq \eta^t_{k+1}$. In this case, the update is
\begin{equation*}
{\bm V}^{t+1}_{k+1,i}={\bm Y}_{k+1,i}^t={\bm U}_{k+1,i}^t+{\bm\Lambda^t_{k+1,i}}/{\rho}.
\end{equation*}
Taking the limit $t\to\infty$ and using the fact that ${\bm U}^{\infty}={\bm V}^{\infty}$, we obtain $\bm\Lambda_{k+1,i}^\infty=0$. Combining with~\eqref{eq:important}, we conclude that $\bm\Lambda_{k+1,i}^\infty=-\nabla_{k+1} \mathcal{F}_i({\bm{\bar{u}}_i^\infty})=0$.

\noindent\underline{\it Case 2:} $\|{\bm Y}^t_{k+1,i}\|_2< \eta^t_{k+1}$. Here, the update yields ${\bm V}^{t+1}_{k+1,i}=0$. Taking the limit $t\to\infty$, we obtain ${\bm V}^{\infty}_{k+1,i}={\bm U}^{\infty}_{k+1,i}=0$. Furthermore, from Case 1 we have $\eta_{k+1}^\infty = \eta_{k+1}$. Therefore,
\begin{equation*}
\|{\bm Y}^{\infty}_{k+1,i}\|_2=\|{\bm U}_{k+1,i}^\infty+{\bm\Lambda^\infty_{k+1,i}}/{\rho}\|_2=\|{\bm\Lambda^\infty_{k+1,i}}/{\rho}\|_2\leq\eta_{k+1}.
\end{equation*}
Thus, $\|\nabla_{k+1} \mathcal{F}_i({\bm{\bar{u}}_i^\infty})\|_2=\|\bm\Lambda^\infty_{k+1,i}\|_2\leq\rho\eta_{k+1}\leq L\eta_{k+1}$.

Since the two cases correspond exactly to (II), the proof of Theorem~\ref{thm:stationary} is complete.
\bibliographystyle{IEEEtran}
\bibliography{ref.bib}

@article{xu2018optimal,
	author = {Xu, J. and Wen, C. and Xu, D.},
	journal = {Sci. China Inf. Sci.},
	pages = {1--15},
	publisher = {Springer},
	title = {Optimal control data scheduling with limited controller-plant communication},
	volume = {61},
	year = {2018}}

@article{jiao2022actuator,
	author = {Jiao, J. and Maity, D. and Baras, J. S. and Hirche, S.},
	journal = {IEEE Control Syst. Lett.},
	pages = {7--12},
	publisher = {IEEE},
	title = {Actuator scheduling for linear systems: A convex relaxation approach},
	volume = {7},
	year = {2022}}

@book{karp2010reducibility,
	author = {Karp, R. M.},
	publisher = {Springer},
	title = {Reducibility among Combinatorial Problems},
	year = {2010}}

@article{boyd2011distributed,
	author = {Boyd, Stephen and Parikh, Neal and Chu, Eric and Peleato, Borja and Eckstein, Jonathan and others},
	journal = {Found. Trends Mach. Learn.},
	number = {1},
	pages = {1--122},
	publisher = {Now Publishers, Inc.},
	title = {Distributed optimization and statistical learning via the alternating direction method of multipliers},
	volume = {3},
	year = {2011}}

@article{shi2022cardinality,
	author = {Shi, Zhang-Lei and Li, Xiao Peng and Leung, Chi-Sing and So, Hing Cheung},
	journal = {IEEE Trans. Neural Netw. Learn. Syst.},
	number = {2},
	pages = {2901--2909},
	publisher = {IEEE},
	title = {Cardinality constrained portfolio optimization via alternating direction method of multipliers},
	volume = {35},
	year = {2022}}

@article{zhong2024sparse,
	author = {Zhong, Yuxing and Yang, Nachuan and Huang, Lingying and Shi, Guodong and Shi, Ling},
	journal = {Automatica},
	pages = {111670},
	publisher = {Elsevier},
	title = {Sparse sensor selection for distributed systems: An $l_1$-relaxation approach},
	volume = {165},
	year = {2024}}

@inproceedings{chartrand2008iteratively,
	author = {Chartrand, Rick and Yin, Wotao},
	booktitle = {Proc. IEEE Int. Conf. Acoust. Speech Signal Process.},
	pages = {3869--3872},
	title = {Iteratively reweighted algorithms for compressive sensing},
	year = {2008}}

@article{magnusson2015convergence,
	author = {Magn{\'u}sson, Sindri and Weeraddana, Pradeep Chathuranga and Rabbat, Michael G and Fischione, Carlo},
	journal = {IEEE Trans. Control Netw. Syst.},
	number = {3},
	pages = {296--309},
	publisher = {IEEE},
	title = {On the convergence of alternating direction {Lagrangian} methods for nonconvex structured optimization problems},
	volume = {3},
	year = {2015}}

@article{peng2015proximal,
	author = {Peng, Zheng and Chen, Jianli and Zhu, Wenxing},
	journal = {J. Glob. Optim.},
	number = {4},
	pages = {711--728},
	publisher = {Springer},
	title = {A proximal alternating direction method of multipliers for a minimization problem with nonconvex constraints},
	volume = {62},
	year = {2015}}

@article{shi2012finite,
	author = {Shi, Ling and Yuan, Ye and Chen, Jiming},
	journal = {IEEE Trans. Autom. Control},
	number = {7},
	pages = {1835--1841},
	publisher = {IEEE},
	title = {Finite horizon {LQR} control with limited controller-system communication},
	volume = {58},
	year = {2012}}

@article{hespanha2007survey,
	author = {Hespanha, Joo P and Naghshtabrizi, Payam and Xu, Yonggang},
	journal = {Proc. IEEE},
	number = {1},
	pages = {138--162},
	publisher = {IEEE},
	title = {A survey of recent results in networked control systems},
	volume = {95},
	year = {2007}}

@article{walsh2001scheduling,
	author = {Walsh, Gregory C and Ye, Hong},
	journal = {IEEE Control Syst. Mag.},
	number = {1},
	pages = {57--65},
	publisher = {IEEE},
	title = {Scheduling of networked control systems},
	volume = {21},
	year = {2001}}

@article{pasqualetti2014controllability,
	author = {Pasqualetti, Fabio and Zampieri, Sandro and Bullo, Francesco},
	journal = {IEEE Trans. Control Netw. Syst.},
	number = {1},
	pages = {40--52},
	publisher = {IEEE},
	title = {Controllability metrics, limitations and algorithms for complex networks},
	volume = {1},
	year = {2014}}

@article{summers2014optimal,
	author = {Summers, Tyler H and Lygeros, John},
	journal = {IFAC Proc. Vol.},
	number = {3},
	pages = {3784--3789},
	publisher = {Elsevier},
	title = {Optimal sensor and actuator placement in complex dynamical networks},
	volume = {47},
	year = {2014}}

@article{chan2007state,
	author = {Chan, Ching Chuen},
	journal = {Proc. IEEE},
	number = {4},
	pages = {704--718},
	publisher = {IEEE},
	title = {The state of the art of electric, hybrid, and fuel cell vehicles},
	volume = {95},
	year = {2007}}

@inproceedings{nishida2024sparsity,
	author = {Nishida, Shumpei and Okano, Kunihisa},
	booktitle = {Proc. Eur. Control Conf.},
	organization = {IEEE},
	pages = {3612--3617},
	title = {Sparsity-constrained linear quadratic regulation problem: Greedy approach with performance guarantee},
	year = {2024}}

@article{kishida2018hands,
	author = {Kishida, Masako and Barforooshan, Mohsen and Nagahara, Masaaki},
	journal = {IFAC-PapersOnLine},
	number = {23},
	pages = {355--360},
	publisher = {Elsevier},
	title = {Hands-off control for discrete-time linear systems subject to polytopic uncertainties},
	volume = {51},
	year = {2018}}

@article{beck2013sparsity,
	author = {Beck, Amir and Eldar, Yonina C},
	journal = {SIAM J. Optim.},
	number = {3},
	pages = {1480--1509},
	publisher = {SIAM},
	title = {Sparsity constrained nonlinear optimization: Optimality conditions and algorithms},
	volume = {23},
	year = {2013}}

@article{bertsekas1997nonlinear,
	author = {Bertsekas, Dimitri P},
	journal = {J. Oper. Res. Soc.},
	number = {3},
	pages = {334--334},
	publisher = {Taylor \& Francis},
	title = {Nonlinear programming},
	volume = {48},
	year = {1997}}

@article{mesbah2016stochastic,
	title={Stochastic model predictive control: An overview and perspectives for future research},
	author={Mesbah, Ali},
	journal={IEEE Control Syst. Mag.},
	volume={36},
	number={6},
	pages={30--44},
	year={2016},
	publisher={IEEE}
}

@article{siami2020separation,
	title={A separation theorem for joint sensor and actuator scheduling with guaranteed performance bounds},
	author={Siami, Milad and Jadbabaie, Ali},
	journal={Automatica},
	volume={119},
	pages={109054},
	year={2020},
	publisher={Elsevier}
}

@article{nagahara2015maximum,
	title={Maximum hands-off control: a paradigm of control effort minimization},
	author={Nagahara, Masaaki and Quevedo, Daniel E and Ne{\v{s}}i{\'c}, Dragan},
	journal={IEEE Trans. Autom. Control},
	volume={61},
	number={3},
	pages={735--747},
	year={2015},
	publisher={IEEE}
}

@article{siami2020deterministic,
	title={Deterministic and randomized actuator scheduling with guaranteed performance bounds},
	author={Siami, Milad and Olshevsky, Alexander and Jadbabaie, Ali},
	journal={IEEE Trans. Autom. Control},
	volume={66},
	number={4},
	pages={1686--1701},
	year={2020},
	publisher={IEEE}
}

@article{ballotta2024pointwise,
	title={Pointwise-sparse actuator scheduling for linear systems with controllability guarantee},
	author={Ballotta, Luca and Joseph, Geethu and Thete, Irawati Rahul},
	journal={IEEE Control Syst. Lett.},
	volume={8},
	pages={2361--2366},
	year={2024},
	publisher={IEEE}
}

@article{joseph2020controllability,
	title={Controllability of linear dynamical systems under input sparsity constraints},
	author={Joseph, Geethu and Murthy, Chandra R},
	journal={IEEE Trans. Autom. Control},
	volume={66},
	number={2},
	pages={924--931},
	year={2020},
	publisher={IEEE}
}

@article{pasand2017structural,
	title={Structural properties, LQG control and scheduling of a networked control system with bandwidth limitations and transmission delays},
	author={Pasand, Mohammad Mahdi Share and Montazeri, Mohsen},
	journal={IEEE/CAA J. Autom. Sin.},
	volume={4},
	number={4},
	pages={817--829},
	year={2017},
	publisher={IEEE}
}

@article{ikeda2022sparse,
	title={Sparse control node scheduling in networked systems based on approximate controllability metrics},
	author={Ikeda, Takuya and Kashima, Kenji},
	journal={IEEE Trans. Control of Netw. Syst.},
	volume={9},
	number={3},
	pages={1166--1177},
	year={2022},
	publisher={IEEE}
}

@article{olshevsky2020relaxation,
	title={On a relaxation of time-varying actuator placement},
	author={Olshevsky, Alex},
	journal={IEEE Control Syst. Lett.},
	volume={4},
	number={3},
	pages={656--661},
	year={2020},
	publisher={IEEE}
}

@article{xu2016empirical,
	title={An empirical study of {ADMM} for nonconvex problems},
	author={Xu, Zheng and De, Soham and Figueiredo, Mario and Studer, Christoph and Goldstein, Tom},
	journal={arXiv preprint arXiv:1612.03349},
	year={2016}
}

@book{bertsekas2012dynamic,
	title={Dynamic programming and optimal control: Volume I},
	author={Bertsekas, Dimitri},
	volume={4},
	year={2012},
	publisher={Athena scientific}
}

@article{ma2022smart,
	title={Smart actuation for end-edge industrial control systems},
	author={Ma, Yehan and Wang, Yebin and Di Cairano, Stefano and Koike-Akino, Toshiaki and Guo, Jianlin and Orlik, Philip and Guan, Xinping and Lu, Chenyang},
	journal={IEEE Trans. Autom. Sci. Eng.},
	volume={21},
	number={1},
	pages={269--283},
	year={2022},
	publisher={IEEE}
}

@article{mozaffari2019tutorial,
	title={A tutorial on {UAVs} for wireless networks: Applications, challenges, and open problems},
	author={Mozaffari, Mohammad and Saad, Walid and Bennis, Mehdi and Nam, Young-Han and Debbah, M{\'e}rouane},
	journal={IEEE Commun. Surv. Tutor.},
	volume={21},
	number={3},
	pages={2334--2360},
	year={2019},
	publisher={IEEE}
}

@article{daubechies2010iteratively,
	title={Iteratively reweighted least squares minimization for sparse recovery},
	author={Daubechies, Ingrid and DeVore, Ronald and Fornasier, Massimo and G{\"u}nt{\"u}rk, C Sinan},
	journal={Commun. Pure Appl. Math.},
	volume={63},
	number={1},
	pages={1--38},
	year={2010},
	publisher={Wiley Online Library}
}

@article{motee2008optimal,
  title={Optimal control of spatially distributed systems},
  author={Motee, Nader and Jadbabaie, Ali},
  journal={IEEE Transactions on Automatic Control},
  volume={53},
  number={7},
  pages={1616--1629},
  year={2008},
  publisher={IEEE}
}

\end{document}